\documentclass[12pt]{amsart}

\usepackage{amsfonts}
\usepackage{amsmath}
\newtheorem{thm}{Theorem}[section]

\newtheorem{prop}[thm]{Proposition}
\newtheorem*{prob*}{Problem}
\newtheorem*{thm*}{Theorem}

\theoremstyle{definition}

\newtheorem*{defn*}{Definition}

\newtheorem*{rem*}{Remark}
\numberwithin{equation}{section}

\newcommand{\C}{\mathbb C}
\newcommand{\R}{\mathbb R}

\newcommand{\Z}{\mathbb Z}
\newcommand{\N}{\mathbb N}

\DeclareMathOperator{\const}{const}

\DeclareMathOperator{\otherwise}{otherwise}

\newcommand{\re}{\mathop{\mathrm{Re}}}

\begin{document}
\title[Dynamical correlation functions for   random matrices]
 {\bf{Dynamical correlation functions for  products of random matrices}}

\author{Eugene Strahov}
\address{Department of Mathematics, The Hebrew University of
Jerusalem, Givat Ram, Jerusalem
91904}\email{strahov@math.huji.ac.il}
\commby{}
\begin{abstract}
 We introduce and study a family of random processes with a discrete time related to products of random matrices.
 Such processes are formed by singular values of random matrix products, and the number of factors in a random matrix product  plays a role of a discrete time.
 We consider in detail the case when the (squared) singular values of the initial random matrix  form a polynomial ensemble, and  the initial random matrix  is multiplied by standard complex Gaussian matrices. In this case we show that the random process is a discrete-time determinantal point process. For three special cases (the case when the initial random matrix is a standard complex Gaussian matrix, the case when it is a truncated unitary matrix, or the case when it is a standard complex Gaussian matrix with a source) we compute the dynamical correlations functions explicitly, and find the hard edge scaling limits of the correlation kernels.
The proofs rely on the Eynard-Mehta theorem, and on contour integral representations for the correlation kernels suitable for an asymptotic analysis.
\end{abstract}
\maketitle
\section{Introduction}
One of the most recent  achievements  in Random Matrix Theory is the discovery of the exact solvability of ensembles
defined by products of random independent matrices. In particular, it was shown in Akemann and Burda \cite{Akemann1},  Akemann, Ipsen, and Kieburg \cite{AkemannIpsenKieburg}
that both eigenvalues and singular values of products of complex Ginibre matrices  form determinantal point processes, and that the kernels of these processes can be computed explicitly in terms of special functions.
Different scaling limits of these determinantal point processes and of their extensions were investigated in the papers by
Kuijlaars and  Zhang \cite{KuijlaarsZhang},  Kuijlaars and  Stivigny \cite{KuijlaarsStivigny},
Liu, Wang, and Zhang \cite{LiuWangZhang}, Forrester and Liu \cite{ForresterLiu}.
We refer the reader to the paper by Akemann and Ipsen \cite{AkemannIpsen}  for a survey concerning exactly solvable ensembles defined by products of random matrices.

Such ensembles arise in the theory of wireless communication,
in particular in the problem about transmission of vector-valued signals through a sequence of clusters
of scatterers, see M$\ddot{\mbox{u}}$ller \cite{Muller}. In this problem the transmission of a vector-valued signal between two adjacent clusters of scatterers is modeled by a random matrix, and the transmission of a vector-valued signal through a sequence of clusters of scatterers is modeled by the corresponding product of random independent matrices.
Exact solvability of the relevant random matrix ensemble enables to give formulae
for the expectation of the mutual information, which is an important quantity in wireless telecommunication, see Akemann, Ipsen, and Kieburg \cite{AkemannIpsenKieburg},
and Wei,  Zheng,  Tirkkonen, and Hamalainen \cite{Wei}.

Considering the model described above  we note that the ordinal index of a cluster in the sequence can be associated with a discrete time, and the number of factors in the relevant random matrix product
can be identified with a discrete time as well.
Motivated by this observation, the present paper deals with random processes with a discrete time constructed in terms of products of random matrices as follows.
Recall that complex Ginibre matrices are those
whose entries are i.i.d standard complex Gaussian variables. Assume that at time $k=1$ we are given a random matrix $G_1$
of size $N_1\times N_0$, where $N_1\geq N_0$. Assume further  that squared singular values of $G_1$ form a polynomial ensemble, see Section \ref{SectionProcessRandomInitialConditions} for a precise definition.
Set $X(1)=G_1$.  At time $k=2$ the matrix $G_1$ is multiplied (from the left)
by a new matrix $G_2$, where $G_2$ is a complex Ginibre matrix of size $N_2\times N_1$ (where $N_2\geq N_0$) independent on $G_1$.
At time $k=3$ the matrix $X(2)=G_2G_1$ is again multiplied from the left by a new complex Ginibre matrix $G_3$ of size $N_3\times N_2$
independent on $G_1$, $G_2$. Continuing in this way we obtain a sequence of product matrices $\left\{X(k)\right\}_{k=1}^{m}$
defined by
$
X(k)=G_k\ldots G_2G_1,
$
where $k=1,\ldots, m$, the squared singular values of $G_1$ form a polynomial ensemble, and  $G_2$,$\ldots$,$G_m$ are independent matrices with i.i.d standard complex Gaussian entries.

Note that for each $k=1,\ldots, m $ the matrices $X^*(k)X(k)$ are of the same size $N_0\times N_0$.
We wish to describe  the evolution of the eigenvalues of  $X^*(k)X(k)$. In other words, the problem
is to describe the evolution of the (squared) singular values of  $X(k)$. Clearly, this evolution
can be thought in terms of non-intersecting polygon paths parameterized by  the (discrete) time parameter $k$.

Denote
by  $x_j^k$ the $j$th largest eigenvalue of  $X(k)^*X(k)$.
If matrices $G_k$ are independent, $G_2$,  $\ldots$, $G_m$
are complex Ginibre matrices, and squared singular values of $G_1$ form a polynomial ensemble,
the configuration  $(k,x_j^k)$ is shown to be a determinantal point process on $\left\{1,\ldots,m\right\}\times\R_{>0}$.
This is the main result of the present paper.
For three special cases (when $G_1$ is a complex Ginibre matrix, when $G_1$ is a truncated unitary matrix, and when  $G_1$ is a complex Ginibre
matrix with a source) we compute correlation kernels of the dynamical correlation functions explicitly, and find their hard edge scaling limits.
The proofs rely on the Eynard-Mehta theorem \cite{EynardMehta}, and on contour integral representations of the correlation kernels suitable for
an asymptotic analysis.

Let us mention that  recent works on products of random matrices (in particular, the works mentioned above) deal with the situation when
the number of matrices in products under consideration is fixed. In this article we propose to consider
\textit{sequences} of products of random matrices instead of fixed products. The main idea behind this work is rather similar to that of Johansson and Nordenstam \cite{JohanssonNordenstam}.
Johansson and Nordenstam \cite{JohanssonNordenstam} start with an infinite random matrix picked from the Gaussian Unitary
Ensemble, and consider the eigenvalues of its minors. They show that all such eigenvalues form a determinantal process on $\N\times\R$. Such determinantal processes are called the minor processes.
Here we can start from a  matrix product $G_m\ldots G_1$ , and consider squared singular values of  sub-products $G_k\ldots G_1$, $k=1,\ldots,m$.  As a result we arrive to a family of determinantal processes on $\left\{1,\ldots,m\right\}\times\R_{>0}$, which can be understood as  time-dependent extensions of the processes studied previously.

The minor processes and their generalizations are of great interest in Random Matrix Theory, and are closely related to random tiling models, and to certain percolation models
on the $\Z^+\times\Z^+$ lattice. In particular, it was observed by Johansson and Nordenstam \cite{JohanssonNordenstam} (see also Okounkov and Reshetikhin \cite{okounkov})
that the spectra of the principal minors of a GUE-matrix behave as domino tilings of large size Aztec diamonds.
We refer the reader to Adler, van Moerbeke, Wang \cite{Adler}, Adler,  Chhita,  Johansson, and van Moerbeke \cite{Adler1}, Adler and van Moerbeke \cite{Adler2} and to references therein for a description
of these connections, and for recent results in this area of research.  Many others time-dependent deteriminantal point processes appeared over the past decade
in the context of growth models, and were applied to different problems in numerative combinatorics and statistical physics, see, for example,
Borodin and Gorin \cite{BorodinGorin}, Borodin and Ferrari \cite{BorodinFerrari}

The  novelty of the present paper  is in the very observation that products of random matrices also lead to  time-dependent determinantal point processes. These processes are rather different from those formed by eigenvalues
of minors, have different scaling limits, and were not considered in the literature, to the best of author's knowledge. One should expect that time-dependent determinantal
processes originated from products of random matrices lead to interesting models of statistical mechanics (similar to the minor processes), but we leave
this question  for a future research. On the technical level the main contribution of the present paper is in the demonstration of a relation with the Eynard-Mehta theorem, and of its usefulness
for derivations of correlation functions. This relation not only give new results (explicit formulae for the dynamical correlation functions),
but also provides a new and a unified point of view on the point processes formed by the squared singular values of random matrices.

The paper is organized as follows. In Section \ref{SectionMainResults} we state the main results obtained in
this work. In particular, Theorem \ref{TheoremPolynomialInitialConditions} of Section \ref{SectionProcessRandomInitialConditions}
gives a general formula for the correlation kernel of the determinant process under consideration.
Sections \ref{SectionProcessRandomInitialConditions}-\ref{GinibreSource} provide contour integral representations
and hard edge scaling limits of the correlation kernels corresponding to three special initial conditions mentioned above.
The rest of the paper is devoted to proofs of these results.\\
\textbf{Acknowledgements.}
I am   deeply grateful to  Dong Wang for a number of remarks and suggestions, and for bringing the preprint \cite{Claeys}
to my attention.
\section{Main results}\label{SectionMainResults}
\subsection{Ginibre product process with
polynomial ensemble initial conditions}\label{SectionProcessRandomInitialConditions}
Let $G_1$,$\ldots$,$G_m$ be random matrices, and assume that each matrix $G_k$, $k\in\left\{1,\ldots,m\right\}$, is of size $N_k\times N_{k-1}$.
Set
$$
X(k)=G_k\ldots G_1,\;\; k\in\left\{1,\ldots,m\right\}.
$$
If $n=N_0$, then for each $k$, $k\in\left\{1,\ldots,m\right\}$, $X(k)^*X(k)$ are random matrices of the same size $n\times n$.
Furthermore, set
$$
\nu_j=N_j-N_0,\;\;  j\in\left\{0,\ldots,m\right\},
$$
and assume that these parameters are positive. Denote by $x_j^k$ the $j$th largest eigenvalue of  $X(k)^*X(k)$. The configuration of all these eigenvalues,
$$
\left\{(k,x_j^k)\vert k=1,\ldots,m;\; j=1,\ldots,n\right\},
$$
form a point process on $\left\{1,\ldots,m\right\}\times \R_{>0}$. The correlation functions of this point process,\footnote{We refer the reader to Johansson \cite{Johansson}
for a general discussion of correlation functions for point processes.}
$$
\varrho_l:\;\;\left(\left\{1,\ldots,m\right\}\times \R_{>0}\right)^l\rightarrow\R,
$$
will be called the \textit{dynamical correlation functions for the products of random matrices $G_1$,$\ldots$,$G_m$}. Also, we say that the random matrix $G_1$ determines the initial conditions of the product matrix process under considerations.
The random matrices $G_2,\ldots,G_m$ will be called the transition matrices.

In what follows we always assume that \textit{the matrices $G_1,\ldots,G_m$ are independent, and the transition matrices
$G_2,\ldots,G_m$ are  complex Ginibre matrices.}
We consider the case when the squared
singular values $\left(x_1^1,\ldots,x_n^1\right)$ of the matrix $G_1$ form a polynomial
ensemble in the sense of Kuijlaars and Stivigny \cite{KuijlaarsStivigny}. This means that the vector
$\left(x_1^1,\ldots,x_n^1\right)$ has density\footnote{Kuijlaars and Stivigny \cite{KuijlaarsStivigny} defined the polynomial ensemble by the expression
\begin{equation}\label{PolynomialEnsembleKS}
\const\cdot\triangle\left(x_1,\ldots,x_n\right)\det\left[f_{i}(x_j)\right]_{i,j=1}^n.
\end{equation}
For a computational convenience, we use equation (\ref{PolynomialEnsemble}) instead of equation
(\ref{PolynomialEnsembleKS}).}
\begin{equation}\label{PolynomialEnsemble}
\const\cdot\triangle\left(x_1^1,\ldots,x_n^1\right)\det\left[f_{i}(x_j^1)\right]_{i,j=1}^n\prod\limits_{j=1}^n\left(x_j^1\right)^{\nu_1}e^{-x_j^1},
\end{equation}
where $\triangle\left(x_1^1,\ldots,x_n^1\right)=\prod\limits_{1\leq j<k\leq n}(x_k^1-x_j^1)$ is the Vandermonde determinant.
We assume that given functions  $f_1$, $\ldots$, $f_{n}$ are such that expression (\ref{PolynomialEnsemble}) can be understood as a probability density.
We obtain  a point process on $\left\{1,\ldots,m\right\}\times \R_{>0}$. Let us refer to this process as to the \textit{Ginibre product process with
polynomial ensemble initial conditions}.  The density of this product process is that of the vector
$$
\underline{x}=\left(x^1,\ldots,x^m\right)\in\left(\R_{>0}^n\right)^m,
$$
where $x^1=\left(x_1^1,\ldots,x_n^1\right)$ are the squared singular values of $G_1$, and
$x^r=\left(x_1^r,\ldots,x_n^r\right)$, $r\in\left\{2,\ldots,m\right\}$, are the squared singular values of $G_r\ldots G_1$.

To present our result for the Ginibre product process with
polynomial ensemble initial conditions
introduce the Meijer G-function (see e.g. Luke \cite{Luke})
\begin{equation}\label{G-def}
G_{p,q}^{m,n}\left(\begin{array}{cccc}
                                                                a_1, & a_2, & \ldots, & a_p \\
                                                                b_1, & b_2, & \ldots, & b_q
                                                              \end{array}
\biggl|z\right)=\frac{1}{2\pi i}\int\limits_{C}\frac{\prod_{j=1}^m\Gamma(b_j-s)\prod_{j=1}^n\Gamma(1-a_j+s)}{\prod_{j=m+1}^q\Gamma(1-b_j+s)\prod_{j=n+1}^p\Gamma(a_j-s)}z^sds.
\end{equation}
Here $0\leq m\leq q$, $0\leq n\leq p$, and
an empty product is interpreted as unity. The contour of integration $C$ depends on the location of the poles of the Gamma functions, and we refer to the NIST handbook \cite{NIST} for details on the different possibilities.
\begin{thm}\label{TheoremPolynomialInitialConditions}
The Ginibre product process with polynomial ensemble initial conditions is a determinantal point process on $\left\{1,\ldots,m\right\}\times\R_{>0}$. Its correlation kernel
can be written as
\begin{equation}\label{MainGeneralFormula}
\begin{split}
K_{n,m}(r,x;s,y)=&-\frac{1}{x}G^{s-r,0}_{0,s-r}\left(\begin{array}{ccc}
                - \\
               \nu_{r+1},\ldots,\nu_s
             \end{array}\biggl|\frac{y}{x}\right)\mathbf{1}_{s>r}\\
&+\sum\limits_{i,j=1}^n\phi_{r,m+1}(x,i)\left(A^{-1}\right)_{i,j}\phi_{0,s}(j,y).
\end{split}
\end{equation}
In the formula just written above the function $\phi_{r,m+1}(x,i)$ is defined as follows. If $m=1$, then
$\phi_{r,m+1}(x,i)=x^{i-1}$. For $m\geq 2$ we set
$$
\phi_{r,m+1}(x,i)=\left\{
                           \begin{array}{ll}
                             x^{i-1}\Gamma(\nu_{r+1}+i)\ldots\Gamma(\nu_m+i), & r=1,\ldots,m-1,  \\
                             x^{i-1}, & r=m.
                           \end{array}
                         \right.
$$
The matrix $A=\left(a_{i,j}\right)_{i,j=1}^n$ in the formula for the correlation kernel
is given by
\begin{equation}\label{GeneralFormulaAij}
a_{i,j}=\Gamma(\nu_2+j)\ldots\Gamma(\nu_m+j)\int\limits_0^{\infty}f_i(t)e^{-t}t^{\nu_1+j-1}dt,\;\; i,j\in\left\{1,\ldots,n\right\},
\end{equation}
and the functions
$
\phi_{0,s}
$
are given by
\begin{equation}
\phi_{0,s}(j,y)=\left\{
                  \begin{array}{ll}
                    f_j(y)y^{\nu_1}e^{-y}, & s=1, \\
                   \int\limits_0^{\infty}f_j(t)t^{\nu_1-1}e^{-t}
G^{s-1,0}_{0,s-1}\left(\begin{array}{ccc}
                - \\
               \nu_2,\ldots,\nu_s
             \end{array}\biggl|\frac{y}{t}\right)dt, & s=2,\ldots,m.\\

                  \end{array}
                \right.
\end{equation}
\end{thm}
If the squared  singular values of $G_1$ form a polynomial ensemble, and $G_2,\ldots,G_m$  are independent Ginibre matrices,
then Theorem  \ref{TheoremPolynomialInitialConditions} implies that the dynamical correlation functions $\varrho_l$ for
the products of random matrices $G_1$,$\ldots$,$G_m$
have the determinantal form,
\begin{equation}\label{DeterminantalFormOfTheCorrelationFunctions}
\varrho_l\left(s_1,x_1;\ldots;s_l,x_l\right)=\det\left(K_{n,m}(s_i,x_i;s_j,x_j)\right)_{i,j=1}^l,\;\; l=1,2,\ldots .
\end{equation}
We note that $K_{n,m}(r,x;s,y)$ has a structure similar to that of the correlation kernel for a continuous version of the Schur
process, see Borodin and P$\acute{\mbox{e}}$ch$\acute{\mbox{e}}$ \cite[Thm. 3]{BorodinPeche}. Both Theorem  \ref{TheoremPolynomialInitialConditions}
and  Borodin and P$\acute{\mbox{e}}$ch$\acute{\mbox{e}}$ \cite[Thm. 3]{BorodinPeche} can be proved using the Eynard-Mehta Theorem \cite{EynardMehta}.

Once the formula for the correlation kernel is derived many probabilistic quantities of interest can be be computed.
In the process interpretation, the evolution of the squared singular values of $X(k)=G_k\cdot\ldots\cdot G_1$ can be thought
of as consisting of $n$ paths parameterized by a discrete time. Given $1\leq l_1<\ldots<l_k\leq m$, and subsets
$\Omega_1$, $\ldots$, $\Omega_k$ of $\R_{>0}$, one quantity of interest is the probability
that for all $i$, $i\in\left\{1,\ldots,k\right\}$, no path passes through $\Omega_i$ at time $l_i$.
It follows from Theorem \ref{TheoremPolynomialInitialConditions} that this probability is expressible as a Fredholm
determinant defined by the kernel $K_{n,m}(r,x;s,y)$. Such Fredholm determinants can be studied by different methods,
in particular in terms of differential equations, see Tracy and Widom \cite{TracyWidomDysonProcesses} for partial differential equations derived in the context of Dyson processes,
and Strahov \cite{strahovdif} for systems of ordinary differential equations derived for singular values of products of independent Ginibre matrices.
Also, paper by Bornemann \cite{Bornemann} gives a systematic numerical treatment of Fredholm determinants arising in the context of Random Matrix Theory, and the methods of this paper might be applicable for the Fredholm determinant defined by the kernel $K_{n,m}(r,x;s,y)$.

Theorem \ref{TheoremPolynomialInitialConditions} gives a general expression for the correlation kernel, which works
for any  sequence of  functions $f_1$, $\dots$, $f_n$, for which (\ref{PolynomialEnsemble})
is a probability density. For some special cases
we can provide more explicit formulae for the correlation kernels, in particular
we can derive  contour integral representations suitable for a subsequent asymptotic analysis.
\subsection{The initial matrix is a complex Ginibre matrix}\label{SectionInitialGinibre}
Assume that $G_1$ is a complex Ginibre matrix. Then the density of the squared singular values of $G_1$
is given by equation (\ref{PolynomialEnsemble}), where the functions $f_1$, $\ldots$, $f_n$ are defined by
\begin{equation}\label{FunctionsFiGinibreMatrix}
f_i(x)=x^{i-1},\;\;i\in\left\{1,\ldots,n\right\}.
\end{equation}
\begin{thm}\label{TheoremGinibreKernel}
When the initial matrix $G_1$ is a complex Ginibre matrix,
the Ginibre product process is a determinantal point process on $\left\{1,\ldots,m\right\}\times\R_{>0}$. Its correlation kernel,
$K_{n,m}(r,x;s,y)$, can be written as\footnote{Note that $\nu_0=0$. This fact will be used throughout the paper.}
\begin{equation}\label{CorrelationKernelGinibre}
K_{n,m}(r,x;s,y)=-\frac{1}{x}G^{s-r,0}_{0,s-r}\left(\begin{array}{ccc}
                - \\
               \nu_{r+1},\ldots,\nu_s
             \end{array}\biggl|\frac{y}{x}\right)\mathbf{1}_{s>r}+\sum\limits_{p=0}^{n-1}P_{r,p}(x)Q_{s,p}(y),
\end{equation}
where the functions $P_{r,p}$ and $Q_{s,p}$ are defined by
\begin{equation}
P_{r,p}(x)=\frac{1}{2\pi i}\oint\limits_{\Sigma_p}\frac{\Gamma(t-p)}{\prod_{j=0}^r\Gamma(t+\nu_j+1)}x^tdt,
\end{equation}
and
\begin{equation}
Q_{s,p}(x)=\frac{1}{2\pi i}\int_{c-i\infty}^{c+i\infty}\frac{\prod_{j=0}^s\Gamma(u+\nu_j)}{\Gamma(u-p)}y^{-u}du,\;\; c>0.
\end{equation}
Here $\Sigma_p$ is a closed contour that encircles $0,1,\ldots,p$ once in the positive direction, and such that $\re t>-1$ for $t\in\Sigma_p$.
\end{thm}
Theorem \ref{TheoremGinibreKernel} enables us to derive a double integral formula for the correlation kernel. Namely, the following Proposition holds true
\begin{prop}\label{PropositionKexactInitialGinibre} We have
\begin{equation}\label{K1}
\begin{split}
K_{n,m}(r,x;s,y)&=-\frac{1}{x}G^{s-r,0}_{0,s-r}\left(\begin{array}{ccc}
                - \\
               \nu_{r+1},\ldots,\nu_s
             \end{array}\biggl|\frac{y}{x}\right)\mathbf{1}_{s>r}\\
             &+\frac{1}{(2\pi i)^2}\int\limits_{-\frac{1}{2}-i\infty}^{-\frac{1}{2}+i\infty}
du\oint\limits_{\Sigma}dt\frac{\prod_{j=0}^s\Gamma(u+\nu_j+1)}{\prod_{j=0}^r\Gamma(t+\nu_j+1)}
\frac{\Gamma(t-n+1)}{\Gamma(u-n+1)}
\frac{x^ty^{-u-1}}{t-u},
\end{split}
\end{equation}
where $\Sigma$ is a closed contour going around $0,1,\ldots,n$ in the positive direction
and such that $\re t>-\frac{1}{2}$ for $t\in\Sigma$.
\end{prop}
If $r=s$, then  the determinant process defined by $K_{n,m}(r,x;s,y)$   is that formed by
squared singular values of products of $r$ independent complex  Ginibre matrices. In this case
the first term in the right-hand side of equation (\ref{K1}) becomes identically equal to zero,
and the second term  coincides with
the correlation kernel  obtained by Kuijlaars and Zhang \cite{KuijlaarsZhang}, see Kuijlaars and Zhang \cite[Prop. 5.1]{KuijlaarsZhang}.
By the same
argument as in Kuijlaars and Zhang \cite{KuijlaarsZhang} we find the scaling limit of $K_{n,m}(r,x;s,y)$
near the origin (hard edge).
\begin{prop}\label{PropositionScalingLimit1} Assume that $x$, $y$ take values in a compact subset of $\R_{>0}$, and assume that the parameters $\nu_1$, $\ldots$, $\nu_m$
are fixed. We have
\begin{equation}
\begin{split}
\underset{n\rightarrow\infty}{\lim}\frac{1}{n}K_{n,m}\left(r,\frac{x}{n};s,\frac{y}{n}\right)&=-\frac{1}{x}G^{s-r,0}_{0,s-r}\left(\begin{array}{ccc}
                - \\
               \nu_{r+1},\ldots,\nu_s
             \end{array}\biggl|\frac{y}{x}\right)\mathbf{1}_{s>r}\\
             &+\frac{1}{(2\pi i)^2}\int\limits_{-\frac{1}{2}-i\infty}^{-\frac{1}{2}+i\infty}
du\oint\limits_{\Sigma}dt\frac{\prod_{j=0}^s\Gamma(u+\nu_j+1)}{\prod_{j=0}^r\Gamma(t+\nu_j+1)}
\frac{\sin\pi u}{\sin\pi t}
\frac{x^ty^{-u-1}}{t-u},
\end{split}
\end{equation}
where $r, s\in\left\{1,\ldots,m\right\}$, and  $\Sigma$  is a contour starting from $+\infty$ in the upper half plane and returning to
$+\infty$ in the lower half plane, leaving $-\frac{1}{2}$ on the left, and encircling $\{0,1,2,\ldots \}$.
\end{prop}
Proposition \ref{PropositionScalingLimit1} implies that the scaling limits of the dynamical correlation
functions do exist at hard edge. These scaling limits characterize the asymptotic behavior of the point process under considerations
as $n\rightarrow\infty$.

\subsection{The initial matrix is a truncation of a random unitary matrix}
Let $U$ be an $l\times l$ Haar distributed unitary matrix, and let $G_1$
be the $(n+\nu_1)\times n$ upper left block of $U$. We assume that $\nu_1\geq 0$, and that
$l\geq 2n+\nu_1$.
\begin{prop}\label{PropositionJiang}
The density of the squared singular values $\left(x_1^1,\ldots,x_n^1\right)$ of $G_1$ can be written as
$$
\const\cdot\triangle\left(x_1^1,\ldots,x_n^1\right)^2\prod\limits_{i=1}^n\left(x_j^{1}\right)^{\nu_1}\left(1-x_j^{1}\right)^{l-2n-\nu_1}\chi_{(0,1)}(x_j^1),
$$
where
$$
\chi_{(0,1)}(x)=\left\{
                  \begin{array}{ll}
                    1, & x\in (0,1), \\
                    0, & \otherwise.
                  \end{array}
                \right.
$$
\end{prop}
\begin{proof}
See Jiang \cite[Prop. 2.1]{Jiang}.
\end{proof}
Thus the squared singular values $\left(x_1^1,\ldots,x_n^1\right)$ of $G_1$ form a polynomial ensemble,
and the functions $f_1,\ldots,f_n$ in equation (\ref{PolynomialEnsemble}) are given by
\begin{equation}\label{FunctionFiTruncation}
f_i(x)=e^{x}x^{i-1}(1-x)^{l-2n-\nu_1}\chi_{(0,1)}(x),\;\; i\in\left\{1,\ldots,n\right\}.
\end{equation}
\begin{thm}\label{TheoremTruncation} When the initial matrix $G_1$ is a truncation of a random unitary matrix (as defined above), the correlation kernel of the Ginibre product process is given by
equation (\ref{CorrelationKernelGinibre}), where the functions $P_{r,p}(x)$, $Q_{s,p}(y)$ are defined by
$$
P_{r,p}(x)=\frac{1}{2\pi i}\oint\limits_{\Sigma_p}\Gamma(t-p)\frac{\Gamma(t+l-2n+1+p)}{\prod_{j=0}^r\Gamma(t+\nu_j+1)}x^tdt,
$$
and
$$
Q_{s,p}(y)=\frac{1}{2\pi i}\int\limits_{c-i\infty}^{c+i\infty}\frac{\prod_{j=0}^s\Gamma(u+\nu_j)}{\Gamma(u-p)\Gamma(u-2n+l+1+p)}y^{-u}du
$$
Here $r, s\in\left\{1,\ldots,m\right\}$, $c>0$, and the contour $\Sigma_p$ is chosen in the same way as in the statement of Theorem \ref{TheoremGinibreKernel}.
\end{thm}
As in the previous case, a double contour integral formula for the correlation function can be derived.
\begin{prop}\label{TheoremExactKernelTruncation} When the initial matrix $G_1$ is a truncation of a random unitary matrix, the correlation kernel,
$K_{n,m}(r,x;s,y)$, can be written as
\begin{equation}\label{KInitialTruncation}
\begin{split}
&K_{n,m}(r,x;s,y)=-\frac{1}{x}G^{s-r,0}_{0,s-r}\left(\begin{array}{ccc}
                - \\
               \nu_{r+1},\ldots,\nu_s
             \end{array}\biggl|\frac{y}{x}\right)\mathbf{1}_{s>r}\\
&+\frac{1}{(2\pi i)^2}\int\limits_{-\frac{1}{2}-i\infty}^{-\frac{1}{2}+i\infty}du\oint\limits_{\Sigma}dt
\frac{\prod_{j=0}^s\Gamma(u+\nu_j+1)}{\prod_{j=0}^r\Gamma(t+\nu_j+1)}\frac{\Gamma(t+1-n)\Gamma(t+l-n+1)}{\Gamma(u+1-n)\Gamma(u+l-n+1)}
\frac{x^ty^{-u-1}}{u-t},
\end{split}
\end{equation}
where  $r, s\in\left\{1,\ldots,m\right\}$, and $\Sigma$ is defined as in Proposition \ref{PropositionKexactInitialGinibre}.
\end{prop}
Formula (\ref{KInitialTruncation}) enables us to find the hard edge scaling limit of the correlation kernel.
After suitable scaling, we obtain the same limiting kernel as in the case where the initial matrix is a complex Ginibre matrix.
\begin{prop}\label{PropositionScalingLimitTruncation} Assume that the parameters $\nu_1$, $\ldots$, $\nu_{m}$ are fixed, and that $l$ grows at least as $2n$.
Then for $x$, $y$ taking values in a compact subset of $\R_{>0}$
\begin{equation}\label{ScalingLimit2}
\begin{split}
&\underset{n\rightarrow\infty}{\lim}\frac{1}{(l-n)n}K_{n,m}\left(r,\frac{x}{(l-n)n};s,\frac{y}{(l-n)n}\right)=-\frac{1}{x}G^{s-r,0}_{0,s-r}\left(\begin{array}{ccc}
                - \\
               \nu_{r+1},\ldots,\nu_s
             \end{array}\biggl|\frac{y}{x}\right)\mathbf{1}_{s>r}\\
             &+\frac{1}{(2\pi i)^2}\int\limits_{-\frac{1}{2}-i\infty}^{-\frac{1}{2}+i\infty}
du\oint\limits_{\Sigma}dt\frac{\prod_{j=0}^s\Gamma(u+\nu_j+1)}{\prod_{j=0}^r\Gamma(t+\nu_j+1)}
\frac{\sin\pi u}{\sin\pi t}
\frac{x^ty^{-u-1}}{t-u},
\end{split}
\end{equation}
where  $r, s\in\left\{1,\ldots,m\right\}$, and $\Sigma$ is chosen in the same way as in Proposition \ref{PropositionScalingLimit1}.
\end{prop}
We note that for $r=s$ the formulae (\ref{KInitialTruncation}) and (\ref{ScalingLimit2})
reduces to those found by Kuijlaars and Stivigny \cite{KuijlaarsStivigny}.


\subsection{The initial matrix is a Ginibre matrix with a source}\label{GinibreSource}
Assume that the random matrix $G_1$ determining the initial conditions of the Ginibre product process is a sum of a complex Ginibre matrix $G$
and a nonrandom matrix $Q$. In this case we use terminology of Forrester and Liu \cite{ForresterLiu}  and say that the initial conditions of the Ginibre product process are
given by a complex Ginibre matrix  with a source.
\begin{prop}\label{PropositionAPlusG}
Assume that the initial matrix $G_1$ is given by
$$
G_1=G+Q,
$$
where $G$ is a complex Ginibre matrix of size $(n+\nu_1)\times n$, and $Q$ is a nonrandom matrix
of the same size whose squared singular values are $\left(q_1,\ldots,q_n\right)$. Then the density of the squared singular values
$\left(x_1^1,\ldots,x_n^1\right)$ of $G_1$ is proportional to
$$
\triangle\left(x_1^1,\ldots,x_n^1\right)\det\left({}_0F_1\left(\nu_1+1;q_ix_j^1\right)\right)_{i,j=1}^n\prod\limits_{j=1}^n\left(x_j^1\right)^{\nu_1}e^{-x_j^1}.
$$
\end{prop}
\begin{proof}
See  Desrosiers and Forrester \cite[Prop. 5]{DesrosiersForrester}, and Forrester \cite{ForresterLogGases}, $\S11.6$.
\end{proof}
Proposition \ref{PropositionAPlusG} implies that the squared singular values of the initial matrix $G_1=G+Q$
form a polynomial ensemble  defined by
equation (\ref{PolynomialEnsemble}), where
$$
f_i(x)=\frac{1}{\Gamma(\nu_1+1)}{}_0F_1\left(\nu_1+1;q_ix\right),\;\; i\in\left\{1,\ldots, n\right\}.
$$
According to Theorem \ref{TheoremPolynomialInitialConditions} the Ginibre product process with the initial matrix $G_1=G+Q$ is again a determinantal point process.
To give an explicit formula for the relevant correlation kernel we follow
to Forrester and Liu \cite{ForresterLiu}, and  introduce  new functions, $\Psi_r(u;x)$ and $\Phi_s(q;y)$. Namely, we set
\begin{equation}
\begin{split}
\Psi_r(u;x)=\frac{1}{(2\pi i)^r\Gamma(\nu_1+1)}\int\limits_{\gamma_1}&dw_1\ldots\int\limits_{\gamma_r}dw_r
\prod\limits_{l=1}^rw_l^{-\nu_l-1}e^{w_l}\\
&\times e^{\frac{x}{w_1\ldots w_r}}{}_0F_1\left(\nu_1+1;-\frac{xu}{w_1\ldots w_r}\right).
\end{split}
\end{equation}
In the formula above  $\gamma_1$, $\ldots$, $\gamma_r$ are paths starting and ending at $-\infty$ and encircling the origin
once in the positive direction.
In addition, set
$$
\Phi_s(q;y)=\frac{1}{2\pi i}\int\limits_{c-\i\infty}^{c+i\infty}dzy^{-z}\vartheta(q;z)\Gamma(\nu_2+z)\ldots\Gamma(\nu_s+z),
$$
where  $c>-\mbox{min}\left(\nu_1,\ldots,\nu_s\right)$, $s\in\left\{1,\ldots,m\right\}$, and
$$
\vartheta(q;z)=\frac{1}{\Gamma(\nu_1+1)}\int\limits_0^{\infty}dtt^{\nu_1+z-1}e^{-t}{}_0F_1\left(\nu_1+1;-qt\right).
$$
\begin{thm}\label{TheoremSourseExact} When the initial matrix is a a Ginibre matrix with a source,
the Ginibre product process is a determinantal point process on $\left\{1,\ldots,m\right\}\times\R_{>0}$. Its correlation kernel,
$K_{n,m}(r,x;s,y)$, can be written as
\begin{equation}\label{knmsourse}
\begin{split}
&K_{n,m}(r,x;s,y)=-\frac{1}{x}G^{s-r,0}_{0,s-r}\left(\begin{array}{ccc}
                - \\
               \nu_{r+1},\ldots,\nu_s
             \end{array}\biggl|\frac{y}{x}\right)\mathbf{1}_{s>r}\\
&+\frac{1}{2\pi i}\int\limits_0^{\infty}du\int\limits_Cdvu^{\nu_1}e^{-u+v}\Psi_r(u;x)\Phi_s(v;y)
\frac{1}{u-v}\prod\limits_{l=1}^n\frac{u+q_l}{v+q_l},
\end{split}
\end{equation}
where $C$ is a counterclockwise contour encircling $-q_1$, $\ldots$, $-q_n$ but not $u$.
\end{thm}
Theorem \ref{TheoremSourseExact} is a time-dependent extension of Forrester and Liu \cite[Prop 1.1]{ForresterLiu}.
Indeed, once $r=s$ Theorem \ref{TheoremSourseExact} gives a correlation kernel for the determinantal process
formed by squared singular values of the random matrix product $G_r\ldots G_2(G+Q)$.
We note that once Theorem \ref{TheoremPolynomialInitialConditions} is established,
the result of  Theorem \ref{TheoremSourseExact}  can be derived using the same arguments
as in Forrester and Liu \cite{ForresterLiu}.

Theorem \ref{TheoremSourseExact} is a starting point for an asymptotic analysis. Because of similarity
of formulae in Theorem \ref{TheoremSourseExact} and in Forrester and Liu \cite[Prop 1.1]{ForresterLiu}
such an asymptotic analysis is a repetition of arguments from Ref. \cite{ForresterLiu}.
In particular, time-dependent  extensions of the critical kernel, and of the deformed critical kernel
of    Forrester and Liu \cite[Thms 1.2, 1.3]{ForresterLiu} are given by Proposition \ref{PropositionScalingLimitSource} below.
\begin{prop}\label{PropositionScalingLimitSource}(A. Subcritical regime) Set $q_1=\ldots=q_n=bn$, and assume that $0<b<1$. For $x$, $y$ in a compact subset of $(0,\infty)$,
and for fixed parameters $\nu_1$, $\ldots$, $\nu_m$
\begin{equation}
\begin{split}
&\underset{n\rightarrow\infty}{\lim}\frac{1}{(l-b)n}K_{n,m}\left(r,\frac{x}{(l-b)n};s,\frac{y}{(l-b)n}\right)=-\frac{1}{x}G^{s-r,0}_{0,s-r}\left(\begin{array}{ccc}
                - \\
               \nu_{r+1},\ldots,\nu_s
             \end{array}\biggl|\frac{y}{x}\right)\mathbf{1}_{s>r}\\
             &+\frac{1}{(2\pi i)^2}\int\limits_{-\frac{1}{2}-i\infty}^{-\frac{1}{2}+i\infty}
du\oint\limits_{\Sigma}dt\frac{\prod_{j=0}^s\Gamma(u+\nu_j+1)}{\prod_{j=0}^r\Gamma(t+\nu_j+1)}
\frac{\sin\pi u}{\sin\pi t}
\frac{x^ty^{-u-1}}{t-u}\\
&=\frac{1}{x}G^{s-r,0}_{0,s-r}\left(\begin{array}{ccc}
                - \\
               \nu_{r+1},\ldots,\nu_s
             \end{array}\biggl|\frac{y}{x}\right)\mathbf{1}_{s>r}\\
&+\int\limits_0^1G^{1,0}_{0,r+2}\left(\begin{array}{cccc}
                - \\
              -\nu_0, -\nu_{1}, \ldots,\nu_r
             \end{array}\biggl|ux\right)G^{s+1,0}_{0,s+2}\left(\begin{array}{cccc}
                - \\
              \nu_0, \nu_{1}, \ldots,\nu_s
             \end{array}\biggl|uy\right)du,
\end{split}
\nonumber
\end{equation}
where $r,s\in\left\{1,\ldots,m\right\}$, and $\Sigma$ is chosen as in  Proposition \ref{PropositionKexactInitialGinibre}.\\
(B. Critical regime) Set
$$
q_1=\ldots=q_n=n(1-\tau/\sqrt{n})^{-1}.
$$
Then for $x$, $y$ in a compact subset of $(0,\infty)$, and for $\tau$ in a compact subset of $\R$
we have
\begin{equation}
\begin{split}
&\underset{n\rightarrow\infty}{\lim}\frac{1}{\sqrt{n}}K_{n,m}\left(r,\frac{x}{\sqrt{n}};s,\frac{y}{\sqrt{n}}\right)=\frac{1}{x}G^{s-r,0}_{0,s-r}\left(\begin{array}{ccc}
                - \\
               \nu_{r+1},\ldots,\nu_s
             \end{array}\biggl|\frac{y}{x}\right)\mathbf{1}_{s>r}\\
&+\frac{1}{2\pi i}\int\limits_0^{\infty}du\int_{i\R}dv\left(\frac{u}{v}\right)^{\nu_1}\frac{e^{-\tau u-\frac{1}{2}u^2+\tau v+\frac{1}{2}v^2}}{u-v}\\
&\times G^{1,0}_{0,r+2}\left(\begin{array}{cccc}
                - \\
              -\nu_0, -\nu_{1}, \ldots,\nu_r
             \end{array}\biggl|ux\right)G^{s+1,0}_{0,s+2}\left(\begin{array}{cccc}
                - \\
              \nu_0, \nu_{1}, \ldots,\nu_s
             \end{array}\biggl|uy\right).
\end{split}
\nonumber
\end{equation}
(C. Deformed critical regime) Set
$$
q_j=\sqrt{n}\sigma_j,\;\;j=1,\ldots,l,\;\; \mbox{and}\;\; q_k=n(1-\tau/\sqrt{n})^{-1},\;\; k=l+1,\ldots,n,
$$
and assume that $\sigma_1$, $\ldots$, $\sigma_l$ are strictly positive.
Then for $x$, $y$ in a compact subset of $(0,\infty)$, and for $\tau$ in a compact subset of $\R$
we have
\begin{equation}
\begin{split}
&\underset{n\rightarrow\infty}{\lim}\frac{1}{\sqrt{n}}K_{n,m}\left(r,\frac{x}{\sqrt{n}};s,\frac{y}{\sqrt{n}}\right)=\frac{1}{x}G^{s-r,0}_{0,s-r}\left(\begin{array}{ccc}
                - \\
               \nu_{r+1},\ldots,\nu_s
             \end{array}\biggl|\frac{y}{x}\right)\mathbf{1}_{s>r}\\
&+\frac{1}{2\pi i}\int\limits_0^{\infty}du\int_{i\R}dv\left(\frac{u}{v}\right)^{\nu_1}\frac{e^{-\tau u-\frac{1}{2}u^2+\tau v+\frac{1}{2}v^2}}{u-v}\\
&\times\prod\limits_{j=1}^l\frac{u+\sigma_j}{v+\sigma_j} G^{1,0}_{0,r+2}\left(\begin{array}{cccc}
                - \\
              -\nu_0, -\nu_{1}, \ldots,\nu_r
             \end{array}\biggl|ux\right)G^{s+1,0}_{0,s+2}\left(\begin{array}{cccc}
                - \\
              \nu_0, \nu_{1}, \ldots,\nu_s
             \end{array}\biggl|uy\right).
\end{split}
\nonumber
\end{equation}
\end{prop}

Finally, let us mention  few directions  for a further research. Theorem 2.1 in a very recent paper by Clayes, Kuijlaars and Wang \cite{Claeys}
suggests that similar methods (as in this paper) can be applied to  processes formed by
eigenvalues of certain sums of random matrices. Also, paper by Clayes, Kuijlaars and Wang \cite{Claeys}
gives a general formula for the correlation kernel of the squared singular values of the product
$G_r\ldots G_2G_1$, where $G_1$ can be any random matrices whose squared singular values form a polynomial
ensemble, while $G_2$, $\ldots$, $G_r$ are independent complex Ginibre matrices, see Corollary 2.16 in Clayes, Kuijlaars and Wang \cite{Claeys}. A question arises
whether there is a time-dependent extension of this formula.
In addition, the methods of the present paper can be extended
to matrix products formed by truncated unitary matrices, and time-dependent extensions of results in Kieburg, Kuijlaars, and Stivigny \cite{KieburgKuijlaarsStivigny},
Forrester and Liu \cite{ForresterLiu}  can be obtained.

The idea to investigate the  dynamical correlation functions comes naturally from considering the signal model of the MIMO
system described in  Wei, Zheng,  Tirkkonen, and Hamalainen \cite{Wei}. In this model the products of independent complex Ginibre matrices arise, and the results
of Section \ref{SectionInitialGinibre} seem to be relevant in the MIMO communication setting. The author is not aware
whether other models considered in this paper (i.e. the products of truncated random matrices, and the products of random matrices with a source)
are useful for the theory of wireless communication.

\section{Measures given by products of determinants and the Eynard-Mehta theorem}
The aim of this Section is to recall one of mechanisms which gives rise to determinantal
point processes, namely the case when the measure defining a point process is given by
products of determinants. We follow Johansson \cite{Johansson}, Section 2.3, and use the same notation and terminology as in Johansson's paper (see also Eynard and Mehta \cite{EynardMehta},  Borodin \cite{BorodinDeterminantalProcesses}, Tracy and Widom \cite{TracyWidomDysonProcesses}).

Let $n, m\geq 1$ be two fixed natural numbers, and let $X_0$, $X_{m+1}$ be two given sets. Let $X$ be a complete separable
metric space, and consider a probability measure on $(X^n)^m$ given by
\begin{equation}\label{ProductDeterminantsMeasure}
\begin{split}
p_{n,m}(\underline{x})d\mu(\underline{x})&=\frac{1}{(n!)^mZ_{n,m}}\det\left(\phi_{0,1}(x_i^0,x_j^1)\right)_{i,j=1}^n\det\left(\phi_{m,m+1}(x_i^m,x_j^{m+1})\right)_{i,j=1}^n\\
&
\times\prod\limits_{r=1}^{m-1}\det\left(\phi_{r,r+1}(x_i^r,x_j^{r+1})\right)_{i,j=1}^nd\mu(\underline{x}).
\end{split}
\end{equation}
In the formula just written above $Z_{n,m}$ is the normalization constant, the functions $\phi_{r,r+1}: X\times X\rightarrow \C$, $r=1,\ldots,m-1$
are given transition functions,  $\phi_{0,1}: X_0\times X\rightarrow \C$ is a given initial transition function,
and $\phi_{m,m+1}: X\times X_{m+1}\rightarrow \C$ is a given final transition function. Also,
$$
\underline{x}=\left(x^1,\ldots,x^m\right)\in\left(X^n\right)^m; \;\; x^r=\left(x^r_1,\ldots,x^r_n\right), r=1,\ldots, m,
$$
the vectors
$$
x^0=(x^0_1,\ldots,x^0_n)\in X_0^n,\;\; x^{m+1}=(x^{m+1}_1,\ldots,x^{m+1}_n)\in X_{m+1}^n,
$$
are fixed initial and final vectors,
and
$$
d\mu(\underline{x})=\prod\limits_{r=1}^m\prod\limits_{j=1}^nd\mu(x_j^r).
$$
Here $\mu$ is a given Borel measure on $X$.
Given two transition functions $\phi$ and $\psi$ set
$$
\phi\ast\psi(x,y)=\int_X\phi(x,t)\psi(t,y)d\mu(t).
$$
The following statement is often referred as the Eynard-Mehta theorem.
\begin{prop}\label{TheoremEynardMehta} The probability measure $p_{n,m}(\underline{x})d\mu(\underline{x})$ defines a determinantal point process on
$\left\{1,\ldots,m\right\}\times X$. The correlation kernel of this determinantal point process, $K_{n,m}(r,x;s,y)$
(where $r,s\in\left\{1,\ldots,m\right\}$, and $x, y\in X)$, is given by the formula
\begin{equation}\label{CorrelationKernelGeneralFormula}
K_{n,m}(r,x;s,y)=-\phi_{r,s}(x,y)+\sum\limits_{i,j=1}^n\phi_{r,m+1}(x,x_i^{m+1})\left(A^{-1}\right)_{i,j}\phi_{0,s}(x_j^0,y).
\end{equation}
The functions $\phi_{r,s}$, and the matrix $A=(a_{i,j})$ (where $i,j=1,\ldots,n$) are defined in terms of transition functions as follows \begin{equation}
\phi_{r,s}(x,y)=\left\{
                  \begin{array}{ll}
                    \left(\phi_{r,r+1}\ast\ldots\ast\phi_{s-1,s}\right)(x,y), & 0\leq r<s\leq m+1, \\
                    0, & r\geq s,
                  \end{array}
                \right.
\end{equation}
and
\begin{equation}\label{aij}
a_{i,j}=\phi_{0,m+1}(x_i^0,x_j^{m+1}).
\end{equation}
\end{prop}
\begin{proof}
See Johansson \cite{Johansson}, Section 2.3.
\end{proof}
\section{Proof of Theorem \ref{TheoremPolynomialInitialConditions}}
Our starting point is Proposition \ref{PropositionKuijlaarsStivigny}, which gives the distribution of squared
singular values of a complex Ginibre matrix multiplied by a constant matrix from the right.
\begin{prop}\label{PropositionKuijlaarsStivigny}
Let $G$ be a complex Ginibre random matrix of size $(n+\nu)\times l$, and let $Q$ be a non-random matrix of size $l\times n$. Let $q_1,\ldots,q_n$
be the squared singular values of $Q$, and assume that $q_j$ is not equal to zero for each $j$, $j=1,\ldots, n$. If $Y=GQ$, then the squared
singular values $y_1$, $\ldots$, $y_n$ of $Y$ have a joint probability density function on $[0,\infty)^n$ proportional to
$$
\frac{\triangle(y_1,\ldots,y_n)}{\triangle(q_1,\ldots,q_n)}\det\left(\frac{y_i^{\nu}}{q_j^{\nu+1}}e^{-\frac{y_i}{q_j}}\right)_{i,j=1}^n.
$$
\end{prop}
\begin{proof}
See Kuijlaars and Stivigny \cite{KuijlaarsStivigny}, Section 2.
\end{proof}
Recall  that the Ginibre product process with polynomial ensemble initial conditions is formed by singular values of the products
$G_k\ldots G_2G_1$, $k\in\left\{1,\ldots,m\right\}$. It is assumed that the density of the squared singular values of $G_1$
is given by equation (\ref{PolynomialEnsemble}), and that $G_k$, $\ldots$, $G_2$ are independent complex Ginibre matrices.
Proposition \ref{PropositionDensity2} shows that the density of this point process can be written as a product of determinants
as in equation (\ref{ProductDeterminantsMeasure}).
\begin{prop}\label{PropositionDensity2} Set $X_0=\left\{1,\ldots,n\right\}$, $X_{m+1}=\left\{1,\ldots,n\right\}$, and $X=\R_{>0}$.
Let $\mu$ be the Lebesgue measure on $\R_{>0}$. Assume that the initial transition
function,
$$
\phi_{0,1}: \left\{1,\ldots,n\right\}\times\R_{>0}\rightarrow\C,
$$
is defined by
$$
\phi_{0,1}(i,x)=f_{i}(x),
$$
and that the final transition function,
$$
\phi_{m,m+1}: \R_{>0}\times\left\{1,\ldots,n\right\}\rightarrow\C,
$$
is defined by
$$
\phi_{m,m+1}(x,i)=x^{i-1}.
$$
In addition, assume that the transition functions $\phi_{r,r+1}: \R_{>0}\times \R_{>0}\rightarrow \C$, $r=1,\ldots,m-1$,
are defined by
\begin{equation}
\begin{split}
\phi_{r,r+1}(x,y)=\left\{
           \begin{array}{ll}
             x^{\nu_1-\nu_2-1}e^{-\frac{y}{x}-x}, & r=1, \\
             x^{\nu_r-\nu_{r+1}-1}e^{-\frac{y}{x}}, & r=2,\ldots,m-2, \\
             x^{\nu_{m-1}-\nu_{m}-1}y^{\nu_m}e^{-\frac{y}{x}}, & r=m-1.
           \end{array}
         \right.
\end{split}
\nonumber
\end{equation}
Then the density of the Ginibre product process with
polynomial ensemble initial conditions is given by formula (\ref{ProductDeterminantsMeasure}).
\end{prop}
\begin{proof}
We use notation introduced in Section \ref{SectionProcessRandomInitialConditions}. To find the  density of the Ginibre product process with
polynomial ensemble initial conditions we need to find the density of the vector $\left(x^1,\ldots,x^m\right)$, where
$x^1=\left(x_1^1,\ldots,x_n^1\right)$ are the squared singular values of $G_1$, and
$x^r=\left(x_1^r,\ldots,x_n^r\right)$, $r\in\left\{2,\ldots,m\right\}$, are the squared singular values of $G_r\ldots G_1$.
Note  the density of $x^1=\left(x_1^1,\ldots,x_n^1\right)$ is given by equation (\ref{PolynomialEnsemble}).
To find the density of $(x^1,x^2)$ we apply Proposition \ref{PropositionKuijlaarsStivigny}.
The result is that this density is proportional to
$$
\frac{\triangle\left(x_1^2,\ldots,x_n^2\right)}{\triangle\left(x_1^1,\ldots,x_n^1\right)}
\det\left(\frac{\left(x_i^2\right)^{\nu_2}}{\left(x_j^1\right)^{\nu_2+1}}e^{-\frac{x_i^2}{x_j^1}}\right)_{i,j=1}^n
\triangle\left(x_1^1,\ldots,x_n^1\right)\det\left[f_i(x_j^1)\right]_{i,j=1}^n\prod\limits_{j=1}^n\left(x_j^1\right)^{\nu_1}e^{-x_j^1}.
$$
Continuing in this way we obtain that the density of  $\left(x^1,\ldots,x^m\right)$ is
proportional to
\begin{equation}
\begin{split}
&\triangle\left(x_1^m,\ldots,x_n^m\right)\det\left[f_i(x_j^1)\right]_{i,j=1}^n
\det\left(\left(x_i^m\right)^{\nu_m}\left(x_j^{m-1}\right)^{\nu_{m-1}-\nu_m-1}e^{-\frac{x_i^m}{x_j^{m-1}}}\right)_{i,j=1}^n\\
&\times\det\left(\left(x_j^{m-2}\right)^{\nu_{m-2}-\nu_{m-1}-1}e^{-\frac{x_i^{m-1}}{x_j^{m-2}}}\right)_{i,j=1}^n
\ldots \det\left(\left(x_j^{2}\right)^{\nu_{2}-\nu_{3}-1}e^{-\frac{x_i^{3}}{x_j^{2}}}\right)_{i,j=1}^n\\
&\times \det\left(\left(x_j^{1}\right)^{\nu_{1}-\nu_{2}-1}e^{-\frac{x_i^{2}}{x_j^{1}}}e^{-x_j^1}\right)_{i,j=1}^n.
\end{split}
\end{equation}
We compare this formula with equation (\ref{ProductDeterminantsMeasure}). The Proposition is proved.
\end{proof}
Proposition \ref{PropositionDensity2} and Proposition \ref{TheoremEynardMehta} say that the Ginibre product process with polynomial ensemble initial conditions
is in fact a determinantal point process on $\left\{1,\ldots,m\right\}\times\R_{>0}$. The correlation kernel of this process is  defined by
equation (\ref{CorrelationKernelGeneralFormula}).  Propositions \ref{PropositionMainTransitionFunction}-\ref{PropositionMatrixAInverse} give the involved transition functions, and the matrix $A$ explicitly.
\begin{prop}\label{PropositionMainTransitionFunction} Assume that $1\leq r<s\leq m$. Then
\begin{equation}
\phi_{r,s}(x,y)=\frac{g_r(x)}{g_s(y)}\frac{1}{x}G^{s-r,0}_{0,s-r}\left(\begin{array}{ccc}
                - \\
               \nu_{r+1},\ldots,\nu_s
             \end{array}\biggl|\frac{y}{x}\right),
\end{equation}
where the functions $g_r$, $1\leq r\leq m$, are defined by
\begin{equation}\label{gr}
g_r(x)=\left\{
         \begin{array}{ll}
           e^{-x}x^{\nu_1}, & r=1, \\
           x^{\nu_r}, & r=2,\ldots,m-1, \\
           1, & r=m.
         \end{array}
       \right.
\end{equation}
\end{prop}
\begin{proof}
 If $1\leq r<s\leq m$, the two-point function $\phi_{r,s}(x,y)$ can be written explicitly as
\begin{equation}
\begin{split}
&\phi_{r,s}(x,y)=\phi_{r,r+1}\ast\ldots\ast\phi_{s-1,s}(x,y)\\
&=\frac{1}{x}\int\limits_0^{\infty}\ldots
\int\limits_{0}^{\infty}e^{-\frac{t_{r+1}}{x}-\frac{t_{r+2}}{t_{r+1}}-\ldots-\frac{t_{s-1}}{t_{s-2}}-\frac{y}{t_{s-1}}}
\left(\frac{t_{r+1}}{x}\right)^{\nu_{r+1}}\left(\frac{t_{r+2}}{t_{r+1}}\right)^{\nu_{r+2}}\ldots\left(\frac{y}{t_{s-1}}\right)^{\nu_{s}}
\frac{dt_{r+1}}{t_{r+1}}\ldots\frac{dt_{s-1}}{t_{s-1}}\\
&\times\left\{
         \begin{array}{ll}
           \frac{x^{\nu_1}e^{-x}}{y^{\nu_s}}, & 1=r<s<m, \\
            \frac{x^{\nu_r}}{y^{\nu_s}}, & 1<r<s<m, \\
            x^{\nu_r}, & 1=r<s=m, \\
            x^{\nu_1}e^{-x}, & r=1, s=m.
         \end{array}
       \right.
\end{split}
\nonumber
\end{equation}
To obtain the formula in the statement of the Proposition use the identity\footnote{Identity (\ref{GElementaryRepresentation}) can be proved inductively
using
$$
\int\limits_0^{\infty}dte^{-t}t^{b_0-1}G^{m,n}_{p,q}\left(\begin{array}{ccc}
                a_1,\ldots,a_p \\
               b_1,\ldots,b_q
             \end{array}\biggl|\frac{s}{t}\right)=G^{m+1,n}_{p,q+1}\left(\begin{array}{ccc}
                a_1,\dots,a_p\\
               b_0,\ldots,b_q
             \end{array}\biggl|s\right),
$$
see Luke \cite{Luke}.}
\begin{equation}\label{GElementaryRepresentation}
\begin{split}
&G^{m,0}_{0,m}\left(\begin{array}{ccc}
                - \\
               b_1,\ldots,b_m
             \end{array}\biggl|\frac{x_m}{x_0}\right)\\
&=\int\limits_0^{\infty}\ldots\int\limits_0^{\infty}\left(\frac{x_1}{x_0}\right)^{b_1}\left(\frac{x_2}{x_1}\right)^{b_2}\ldots\left(\frac{x_m}{x_{m-1}}\right)^{b_m}
e^{-\frac{x_1}{x_0}-\frac{x_2}{x_1}-\ldots-\frac{x_m}{x_{m-1}}}\frac{dx_1}{x_1}\ldots\frac{dx_{m-1}}{x_{m-1}},
\end{split}
\end{equation}
and the definition of the functions $g_r$ (equation (\ref{gr})).
\end{proof}
\begin{prop}\label{PropositionInitialTransitionFunctions}
The  transition functions
$$
\phi_{0,s}: \left\{1,\ldots,n\right\}\times \R_{>0}\rightarrow\C,\;\;s\in\left\{1,\ldots, m\right\},
$$
of the Ginibre product process with polynomial ensemble initial conditions
are given by
\begin{equation}
\phi_{0,s}(j,y)=\left\{
                  \begin{array}{ll}
                    f_j(y), & s=1, \\
                    \frac{1}{g_s(y)}\int\limits_0^{\infty}f_j(t)t^{\nu_1-1}e^{-t}
G^{s-1,0}_{0,s-1}\left(\begin{array}{ccc}
                - \\
               \nu_2,\ldots,\nu_s
             \end{array}\biggl|\frac{y}{t}\right)dt, & s=2,\ldots,m.\\

                  \end{array}
                \right.
\end{equation}
\end{prop}
\begin{proof}
The function $\phi_{0,1}$ is the initial transition function for the Ginibre product process with polynomial ensemble initial conditions defined in the statement of
Proposition \ref{PropositionDensity2}. For $s\in\left\{2,\ldots,m\right\}$ we have
\begin{equation}
\begin{split}
&\phi_{0,s}(j,y)=\phi_{0,1}\ast\phi_{1,s}(j,y)=\int\limits_0^{\infty}f_j(t)\phi_{1,s}(t,y)dt\\
&=\frac{1}{g_s(y)}\int\limits_0^{\infty}f_j(t)t^{\nu_1-1}e^{-t}G^{s-1,0}_{0,s-1}\left(\begin{array}{ccc}
                - \\
               \nu_2,\ldots,\nu_s
             \end{array}\biggl|\frac{y}{t}\right)dt,
\end{split}
\nonumber
\end{equation}
where we have used the formula for $\phi_{1,s}$ derived in Proposition \ref{PropositionMainTransitionFunction}.
\end{proof}
\begin{prop}\label{PropostionTransitionFunctionPhirmPlus1}
The  transition functions
$$
\phi_{r,m+1}: \R_{>0}\times\left\{1,\ldots,n\right\}\rightarrow\C,\;\;r\in\left\{1,\ldots, m\right\},
$$
of the Ginibre product process with polynomial initial conditions
are given by
\begin{equation}
\phi_{r,m+1}(x,i)=\left\{
                  \begin{array}{ll}
                    e^{-x}x^{\nu_1+i-1}\Gamma(\nu_2+i)\ldots\Gamma(\nu_m+i), & r=1, \\
                    x^{\nu_r+i-1}\Gamma(\nu_{r+1}+i)\ldots\Gamma(\nu_m+i), & r=2,\ldots,m-1, \\
                   x^{i-1} , & r=m.
                  \end{array}
                \right.
\end{equation}
\end{prop}
\begin{proof}
If $r=m$, then $\phi_{m,m+1}$ is just the final transition function for the Ginibre product process defined in the statement of
Proposition \ref{PropositionDensity2}. For $r\in\left\{1,\ldots,m-1\right\}$ we have
\begin{equation}
\begin{split}
&\phi_{r,m+1}(x,i)=\phi_{r,r+1}\ast\ldots\ast\phi_{m-1,m}\ast\phi_{m,m+1}(x,i)=\int\limits_0^{\infty}\phi_{r,r+1}\ast\ldots\ast\phi_{m-1,m}(x,t)t^{i-1}dt\\
&=\frac{g_r(x)}{x}\int\limits_0^{\infty}G^{m-r,0}_{0,m-r}\left(\begin{array}{ccc}
                - \\
               \nu_{r+1},\ldots,\nu_m
             \end{array}\biggl|\frac{t}{x}\right)t^{i-1}dt,
\end{split}
\nonumber
\end{equation}
where we have used the formula for $\phi_{r,r+1}\ast\ldots\ast\phi_{m-1,m}$ derived in Proposition \ref{PropositionMainTransitionFunction}.
Recall that $g_r$ is defined by equation (\ref{gr}), and
that the Mellin transform of a Meijer G-function is given by
\begin{equation}
\begin{split}
\int\limits_0^{\infty}t^{u-1}G^{m,n}_{p,q}\left(\begin{array}{ccc}
                a_1,\ldots, a_p \\
               b_1,\ldots, b_q
             \end{array}\biggl|tz\right)dt
=z^{-u}\frac{\prod_{i=1}^m\Gamma(b_i+u)\prod_{i=1}^n\Gamma(1-a_i-u)}{\prod_{i=n+1}^p\Gamma(a_i+u)\prod_{i=m+1}^q\Gamma(1-b_i-u)},
\end{split}
\end{equation}
see Luke \cite{Luke}.
So
\begin{equation}
\begin{split}
\int\limits_0^{\infty}G^{m-r,0}_{0,m-r}\left(\begin{array}{ccc}
                - \\
               \nu_{r+1},\ldots,\nu_m
             \end{array}\biggl|\frac{t}{x}\right)t^{i-1}dt
=x^{i}\Gamma(\nu_{r+1}+i)\ldots\Gamma(\nu_m+i).
\end{split}
\end{equation}
The result follows.
\end{proof}
\begin{prop}\label{PropositionMatrixAInverse}
For the Ginibre product process with polynomial ensemble initial conditions the matrix $A=\left(a_{i,j}\right)_{i,j=1}^n$
in formula (\ref{CorrelationKernelGeneralFormula}) for the correlation kernel is given by
\begin{equation}\label{aij1}
\begin{split}
a_{i,j}=\Gamma(\nu_2+j)\ldots\Gamma(\nu_m+j)\int\limits_0^{\infty}f_i(t)e^{-t}t^{\nu_1+j-1}dt,\;\; i,j\in\left\{1,\ldots,n\right\}.
\end{split}
\end{equation}
\end{prop}
\begin{proof}
The matrix elements $a_{i,j}$ of $A$ are defined by equation (\ref{aij}) in terms of the transition function
$$
\phi_{0,m+1}: \left\{1,\ldots,n\right\}\times \left\{1,\ldots,n\right\}\rightarrow\C.
$$
Let us compute this transition function explicitly. We have
\begin{equation}
\begin{split}
&\phi_{0,m+1}(i,j)=\phi_{0,1}\ast\phi_{1,m+1}(i,j)=\int\limits_{0}^{\infty}\phi_{0,1}(i,t)\phi_{1,m+1}(t,j)dt\\
&=\Gamma(\nu_2+j)\ldots\Gamma(\nu_m+j)\int\limits_{0}^{\infty}f_i(t)e^{-t}t^{\nu_1+j-1}dt.
\end{split}
\nonumber
\end{equation}
Thus
$
a_{i,j}
$
is given by equation (\ref{aij1}).
\end{proof}
Applying Proposition \ref{TheoremEynardMehta} and Propositions \ref{PropositionDensity2}-\ref{PropositionMatrixAInverse}
we see that the correlation kernel of the Ginibre product process with polynomial ensemble initial conditions
can be written as
$$
\check{K}_{n,m}(r,x;s,y)=\frac{g_r(x)}{g_s(y)}K_{n,m}(r,x;s,y),
$$
where $K_{n,m}(r,x;s,y)$ is defined in the statement of Theorem \ref{TheoremPolynomialInitialConditions}, and where the scale factor$g_r(x)/g_s(y)$ emerges
from the formulas for the transition functions derived in Propositions \ref{PropositionDensity2}-\ref{PropostionTransitionFunctionPhirmPlus1}. It follows
from equation (\ref{DeterminantalFormOfTheCorrelationFunctions}) that
 the kernels
$\check{K}_{n,m}(r,x;s,y)$ and $K_{n,m}(r,x;s,y)$ give the same dynamical correlation functions $\varrho_l(s_1,x_1;\ldots;s_l,x_l)$, i.e. these kernels are equivalent.\footnote{Two kernels $K(x,y)$ and $K'(x,y)$ are called equivalent
if $\det\left[K(x_i,x_j]\right)_{i,j=1}^m=\det\left[K'(x_i,x_j)\right]_{i,j=1}^m$, for any $m=1,2,\ldots$, for example $K'(x,y)=(f(x)/f(y))K(x,y)$. Thus two equivalent kernels define the same correlation functions.}
Theorem \ref{TheoremPolynomialInitialConditions} is proved. \qed
\section{Proof of Theorem \ref{TheoremGinibreKernel}, Proposition \ref{PropositionKexactInitialGinibre}, and Proposition \ref{PropositionScalingLimit1}}
In order to prove Theorem \ref{TheoremGinibreKernel} we use the formulae for the correlation kernel  in Theorem \ref{TheoremPolynomialInitialConditions}.
Taking into account that the functions $f_i(x)$ are given by equation (\ref{FunctionsFiGinibreMatrix}), we immediately obtain explicit formulae for the entries
$a_{i,j}$ of the matrix $A$, and for the initial transition functions $\phi_{0,s}(j,s)$. Namely,
$$
a_{i,j}=\Gamma\left(\nu_1+i+j-1\right)\Gamma(\nu_2+j)\ldots\Gamma\left(\nu_m+j\right),\;\;\;i,j\in\left\{1,\ldots,n\right\},
$$
and
\begin{equation}\label{Phi0s1}
\phi_{0,s}(j,y)=
G^{s,0}_{0,s}\left(\begin{array}{cccc}
                - \\
             \nu_1+j-1, \nu_2, \ldots, \nu_s
             \end{array}\biggl|y\right),
\end{equation}
where $j\in\left\{1,\ldots,n\right\}$, and $s\in\left\{1,\ldots,m\right\}$.
Next we use the fact that the inverse $\left(\alpha_{k,l}\right)_{k,l=0}^{N-1}$
of the Hankel matrix
$$
\left(h_{k+l}\right)_{k,l=0}^{N-1},\;\; h_k=\Gamma(k+\nu+1),
$$
is given by\footnote{Recall that the Pochhammer symbol, $(a)_l$, is defined by
$$
(a)_l=a(a+1)\ldots(a+l-1).
$$}
$$
\alpha_{k,l}=\sum\limits_{p=0}^{N-1}\frac{\Gamma(\nu+p+1)}{\left(\Gamma(\nu+1)\right)^2p!}\frac{(-p)_{k}(-p)_{l}}{(\nu_1+1)_{k}(\nu_1+1)_{l}k!l!},
$$
see, for example, Akemann and Strahov \cite{AkemannStrahov}, Section 5.
This gives
\begin{equation}\label{MatrixInverse1}
\begin{split}
\left(A^{-1}\right)_{i,j}&=\frac{1}{\Gamma(\nu_2+i)\ldots\Gamma(\nu_m+i)}\\
&\times\sum\limits_{p=0}^{n-1}\frac{\Gamma(\nu_1+p+1)}{\left(\Gamma(\nu_1+1)\right)^2p!}\frac{(-p)_{i-1}(-p)_{j-1}}{(\nu_1+1)_{i-1}(\nu_1+1)_{j-1}}\frac{1}{(i-1)!(j-1)!},
\end{split}
\end{equation}
where $i, j\in\left\{1,\ldots,n\right\}$.
Proposition \ref{TheoremEynardMehta} says that the correlation kernel
can be represented as a sum of two terms.
Proposition \ref{PropositionMainTransitionFunction} gives the first term for the correlation kernel.
Denote the second term of the correlation kernel by $\widetilde{K}_{n,m}(r,x;s,y)$.
The matrix $A^{-1}$ is given explicitly by equation (\ref{MatrixInverse1}), so we can write
\begin{equation}\label{KCheckPCheckQ}
\begin{split}
\widetilde{K}_{n,m}(r,x;s,y)&=\sum\limits_{i,j=1}^n\phi_{r,m+1}\left(x,i\right)\left(A^{-1}\right)_{i,j}\phi_{0,s}(j,y)\\
&=\sum\limits_{p=0}^{n-1}\check{P}_{r,p}(x)\check{Q}_{s,p}(y)\frac{\Gamma(\nu_1+p+1)}{\left(\Gamma(\nu_1+1)\right)^2p!},
\end{split}
\end{equation}
where
\begin{equation}
\check{P}_{r,p}(x)=\sum\limits_{i=0}^{p}\phi_{r,m+1}\left(x,i+1\right)\frac{(-p)_i}{\Gamma(\nu_2+i+1)\ldots\Gamma(\nu_m+i+1)(\nu_1+1)_ii!},
\end{equation}
and where
\begin{equation}
\check{Q}_{s,p}(y)=\sum\limits_{j=0}^{p}\phi_{0,s}\left(j+1,y\right)\frac{(-p)_j}{(\nu_1+1)_jj!}.
\end{equation}
The transition function $\phi_{r,m+1}$ can be found as  in Proposition \ref{PropostionTransitionFunctionPhirmPlus1},  and the transition function $\phi_{0,s}$
is given by equation (\ref{Phi0s1}). We find
\begin{equation}
\check{P}_{r,p}(x)=\Gamma(\nu_1+1)p!\sum\limits_{i=0}^{p}\frac{(-1)^{p-i}}{(p-i)!i!}\frac{x^i}{\Gamma(\nu_1+i+1)\ldots\Gamma(\nu_r+i+1)},
\end{equation}
and
\begin{equation}\label{CheckQ}
\check{Q}_{s,p}(y)=\Gamma(\nu_1+1)p!\sum\limits_{j=0}^p\frac{(-1)^{p-j}}{(p-j)!j!\Gamma(\nu_1+1+j)}
G^{s,0}_{0,s}\left(\begin{array}{cccc}
                - \\
               \nu_{1}+j,\nu_2\ldots,\nu_s
             \end{array}\biggl|y\right).
\end{equation}
The residue calculations give
\begin{equation}\label{CheckPContourIntegral}
\check{P}_{r,p}(x)=\frac{\Gamma(\nu_1+1)p!}{2\pi i}\oint\limits_{\Sigma_p}\frac{\Gamma(t-p)}{\prod_{j=0}^r\Gamma(t+\nu_j+1)}x^tdt,
\end{equation}
where $\Sigma_p$ is defined in the same way as in the statement of Theorem \ref{TheoremGinibreKernel}.
To find a contour integral representation for $\check{Q}_{s,p}$ we use the formula
\begin{equation}
G^{s,0}_{0,s}\left(\begin{array}{cccc}
                - \\
               \nu_{1}+j,\nu_2\ldots,\nu_s
             \end{array}\biggl|y\right)=\frac{1}{2\pi i}\int_{c-i\infty}^{c+i\infty}\Gamma(u+\nu_1+j)\Gamma(u+\nu_2)\ldots\Gamma(u+\nu_s)y^{-u}du,
\end{equation}
where $c>0$.
Inserting the contour integral representation just written above into equation (\ref{CheckQ}) we find
\begin{equation}\label{CheckQ1}
\begin{split}
&\check{Q}_{s,p}(y)=\Gamma(\nu_1+1)p!\\
&\times\frac{1}{2\pi i}\int_{c-i\infty}^{c+i\infty}\Gamma(u+\nu_2)\ldots\Gamma(u+\nu_s)
\left(\sum\limits_{j=0}^p\frac{(-1)^{p-j}\Gamma(u+\nu_1+j)}{(p-j)!j!\Gamma(\nu_1+1+j)}\right)y^{-u}du.
\end{split}
\end{equation}
The sum inside the integral in equation (\ref{CheckQ1}) can be rewritten as
\begin{equation}
\begin{split}
&\sum\limits_{j=0}^p\frac{(-1)^{p-j}\Gamma(u+\nu_1+j)}{(p-j)!j!\Gamma(\nu_1+1+j)}=\frac{1}{p!}\frac{\Gamma(u+\nu_1)}{\Gamma(1+\nu_1)}
\sum\limits_{j=0}^p\frac{(-p)_j(u+\nu_1)_j}{(1+\nu_1)_j}\frac{1}{j!}\\
&=\frac{1}{p!}\frac{\Gamma(u+\nu_1)}{\Gamma(1+\nu_1)}{}_2F_1\left(-p,u+\nu_1;1+\nu_1;1\right)\\
&=\frac{1}{p!}\frac{\Gamma(u+\nu_1)}{\Gamma(1+\nu_1)}\frac{(1-u)_p}{(1+\nu_1)_p}=\frac{1}{p!}\frac{\Gamma(u+\nu_1)}{\Gamma(1+\nu_1+p)}\frac{\Gamma(u)}{\Gamma(u-p)},
\end{split}
\nonumber
\end{equation}
where we have used the formula (the Chu-Vandermonde sum)
$$
{}_2F_1\left(-n,b;c;1\right)=\frac{(c-b)_n}{(c)_n},
$$
see Ismail \cite{Ismail}, Section 1.4.
So
\begin{equation}\label{CheckQContourIntegral}
\begin{split}
\check{Q}_{s,p}(y)=\frac{\Gamma(\nu_1+1)}{2\pi i\Gamma(1+\nu_1+p)}\int_{c-i\infty}^{c+i\infty}\frac{\prod_{j=0}^s\Gamma(u+\nu_j)}{\Gamma(u-p)}y^{-u}du.
\end{split}
\end{equation}
Equations (\ref{KCheckPCheckQ}), (\ref{CheckPContourIntegral}), and (\ref{CheckQContourIntegral})
imply that the correlation kernel can be written as in the statement of Theorem \ref{TheoremGinibreKernel}.
Theorem \ref{TheoremGinibreKernel} is proved.

Once Theorem \ref{TheoremGinibreKernel} is established, Proposition \ref{PropositionKexactInitialGinibre}, and Proposition \ref{PropositionScalingLimit1}
can be proved by repetitions of arguments from Kuijlaars and Zhang \cite{KuijlaarsZhang}, Section 5. \qed

\section{Proof of Theorem \ref{TheoremTruncation}, Proposition \ref{TheoremExactKernelTruncation}, and Proposition \ref{PropositionScalingLimitTruncation}}
By Theorem \ref{TheoremPolynomialInitialConditions} and Proposition \ref{PropositionJiang} the Ginibre product process with initial conditions defined by a truncation of a random unitary matrix is
a determinantal point process. The correlation kernel of this determinantal process is given implicitly by formula (\ref{MainGeneralFormula}). To write the correlation kernel explicitly we need to find  explicit expressions for the matrix $A$ and its inverse.
Using equations (\ref{GeneralFormulaAij}) and (\ref{FunctionFiTruncation}) we find that in the case under considerations
the matrix $A=\left(a_{i,j}\right)_{i,j=1}^n$ is defined by
\begin{equation}
a_{i,j}=\Gamma(\nu_2+j)\ldots\Gamma(\nu_m+j)\int\limits_0^{1}t^{i+j-2+\nu_1}(1-t)^{l-2n-\nu_1}dt,\;\; i,j\in\left\{1,\ldots,n\right\}.
\end{equation}
The integral in the equation just written above can be computed using the formula
$$
\int\limits_0^1t^{\alpha}(1-t)^{\beta}dt=\frac{\Gamma(\alpha+1)\Gamma(\beta+1)}{\Gamma(\alpha+\beta+2)}.
$$
The result is
\begin{equation}
a_{i,j}=\Gamma(\nu_2+j)\ldots\Gamma(\nu_m+j)\frac{\Gamma(i+j-1+\nu_1)\Gamma(l-2n-\nu_1+1)}{\Gamma(i+j+l-2n)}.
\nonumber
\end{equation}
We need to find the inverse of the matrix $\check{A}=\left(\check{a}_{i,j}\right)_{i,j=1}^n$ defined by
$$
\check{a}_{i,j}=\frac{\Gamma(i+j-1+\nu_1)}{\Gamma(i+j+l-2n)},\;\;\;\; i,j\in\left\{1,\ldots,n\right\}.
$$
\begin{prop}\label{PropositionZhangChenTheorem10} For $N=1,2,\ldots$, and $-\alpha,-\beta\in\C\setminus\mathbb{N}$ the matrix
$$
\left(\frac{(\alpha+1)_{i+j}}{(\alpha+\beta+2)_{i+j}}\right)_{i,j=0}^{N-1}
$$
is invertible  and its inverse matrix $\left(\gamma_{i,j}\right)_{i,j=0}^{N-1}$ is given by
\begin{equation}
\begin{split}
&\gamma_{i,j}=
\frac{(-1)^{i+j}\left(\alpha+\beta+1\right)_i\left(\alpha+\beta+1\right)_j}{\left(\alpha+1\right)_i\left(\alpha+1\right)_j\left(\alpha+\beta+1\right)}\\
&\times\sum\limits_{p=0}^{N-1}\frac{(2p+\alpha+\beta+1)(\alpha+1)_pp!}{(\alpha+\beta+1)_p(\beta+1)_p(p-i)!i!(p-j)!j!}
\left(\alpha+\beta+i+1\right)_p\left(\alpha+\beta+j+1\right)_p.
\end{split}
\end{equation}
\end{prop}
\begin{proof}
See Theorem 10 in  Zhang and Chen \cite{ZhangChen}.
\end{proof}
We apply Proposition \ref{PropositionZhangChenTheorem10}, and find
\begin{equation}
A^{-1}=\left(\gamma_{i,j}\right)_{i,j=1}^n,
\end{equation}
where
\begin{equation}
\begin{split}
&\gamma_{i,j}=\frac{1}{\Gamma(\nu_2+i)\ldots\Gamma(\nu_m+i)\Gamma(l-2n-\nu_1+1)}\\
&\times\frac{\Gamma(l-2n+2)}{\Gamma(\nu_1+1)}\left(-1\right)^{i+j}\frac{\left(l-2n+1\right)_{i-1}\left(l-2n+1\right)_{j-1}}{(\nu_1+1)_{i-1}(\nu_1+1)_{j-1}\left(l-2n+1\right)}\\
&\times\sum\limits_{p=0}^{n-1}\frac{\left(2p+l-2n+1\right)\left(\nu_1+1\right)_p}{\left(l-2n+1\right)_p\left(l-2n-\nu_1+1\right)_p}
\frac{p!\left(l-2n+i\right)_p\left(l-2n+j\right)_p}{(p-i+1)!(i-1)!(p-j+1)!(j-1)!}.
\end{split}
\end{equation}
To compute the transition function $\phi_{0,s}(j,y)$
consider first the case when $s\in\left\{2,\ldots,m\right\}$. In this case we have
\begin{equation}
\phi_{0,s}(j,y)=\int\limits_0^1t^{\nu_1+j-2}(1-t)^{l-2n+\nu_1}G^{s-1,0}_{0,s-1}\left(\begin{array}{cccc}
                -\\
               \nu_2,\ldots,\nu_s
             \end{array}\biggl|\frac{y}{t}\right)dt.
\end{equation}
 We use the formula
$$
G^{m,n}_{p,q}\left(\begin{array}{ccc}
                a_1,\ldots,a_p\\
               b_1,\ldots,b_q
             \end{array}\biggl|x^{-1}\right)=G^{n,m}_{q,p}\left(\begin{array}{ccc}
               1-b_1,\ldots,1-b_q\\
               1-a_1,\ldots,1-a_p
             \end{array}\biggl|x\right),
$$
see, for example, Luke \cite{Luke},
to rewrite the transition function $\phi_{0,s}(j,y)$ as
\begin{equation}
\phi_{0,s}(j,y)=\int\limits_0^1t^{\nu_1+j-2}(1-t)^{l-2n+\nu_1}G^{0,s-1}_{s-1,0}\left(\begin{array}{cccc}
                1-\nu_2,\ldots,1-\nu_s\\
               -
             \end{array}\biggl|\frac{t}{y}\right)dt.
\end{equation}
This integral above can be computed, and the result is
\begin{equation}\label{InitialTransitionFunctionTruncation}
\phi_{0,s}(j,y)=\Gamma(l-2n-\nu_1+1)G^{s,0}_{1,s}\left(\begin{array}{cccc}
                l-2n+j\\
               \nu_1+j-1,\nu_2,\ldots,\nu_s
             \end{array}\biggl|y\right),
\end{equation}
where $s\in\left\{2,\ldots,m\right\}$, and $j\in\left\{1,\ldots,n\right\}$.
To see that the formula for $\phi_{0,s}(j,y)$ just written above  holds true in the case $s=1$ as well use the identity
$$
\Gamma(l-2n-\nu_1+1)G^{1,0}_{1,1}\left(\begin{array}{c}
               l-2n+j\\
               \nu_1+j-1
             \end{array}\biggl|y\right)=y^{\nu_1+j-1}(1-y)^{l-2n-\nu_1}\chi_{(0,1)}(y).
$$
Now we are ready to derive the formula for the correlation kernel stated in Theorem \ref{TheoremTruncation}.
We start from equation (\ref{MainGeneralFormula}). The second term in equation (\ref{MainGeneralFormula}) can be written as
\begin{equation}\label{KnmTruncated}
\begin{split}
\widetilde{K}_{n,m}\left(r,x;s,y\right)&=\frac{\Gamma(l-2n+1)}{\Gamma(\nu_1+1)\Gamma(l-2n-\nu_1+1)}\\
&\times\sum\limits_{p=0}^{n-1}\frac{p!(\nu_1+1)_p(2p+l-2n+1)}{(l-2n+1)_p(l-2n-\nu_1+1)_p}
\check{P}_{r,p}(x)\check{Q}_{s,p}(y),
\end{split}
\end{equation}
where
\begin{equation}\label{PrpTruncation}
\check{P}_{r,p}(x)=\sum\limits_{i=0}^p\frac{(-1)^{p-i}}{(p-i)!i!}\frac{(l-2n+1)_i(l-2n+i+1)_p}{(\nu_1+1)_i\Gamma(\nu_2+i+1)\ldots\Gamma(\nu_m+i+1)}\phi_{r,m+1}(x,i+1),
\end{equation}
and
\begin{equation}\label{QspTruncation}
\check{Q}_{s,p}(y)=\sum\limits_{j=0}^p\frac{(-1)^{p-j}}{(p-j)!j!}\frac{(l-2n+1)_j(l-2n+j+1)_p}{(\nu_1+1)_j}\phi_{0,s}(j+1,y).
\end{equation}
Let us derive a contour integral representation for the functions $P_{r,p}(x)$.
We have
$$
\phi_{r,m+1}(x,i+1)=\Gamma(\nu_{r+1}+i+1)\ldots\Gamma(\nu_m+i+1)x^i,
$$
so equation (\ref{PrpTruncation}) gives
\begin{equation}
\check{P}_{r,p}(x)=\frac{\Gamma(\nu_1+1)}{\Gamma(l-2n+1)}\sum\limits_{i=0}^p\frac{(-1)^{p-i}}{(p-i)!i!}\frac{\Gamma(l-2n+i+1+p)}{\Gamma(\nu_1+i+1)\ldots\Gamma(\nu_r+i+1)}x^i.
\end{equation}
It can be checked by residue calculations that $P_{r,p}(x)$ can be written as a contour integral
\begin{equation}\label{PrpTruncation1}
\check{P}_{r,p}(x)=\frac{\Gamma(\nu_1+1)}{\Gamma(l-2n+1)}\int\limits_{\Sigma_p}\frac{\Gamma(t-p)\Gamma(t+l-2n+1+p)}{\prod_{j=0}^r\Gamma(t+\nu_j+1)}x^tdt,
\end{equation}
where $\Sigma_p$ is a closed contour that encircles $0,1,\ldots,p$ once in the positive direction, and such that $\re t>-1$ for $t\in\Sigma_p$.

Next, let us obtain an integral representation for the functions $\check{Q}_{s,p}(y)$.  The initial transition function $\phi_{0,s}(j+1,y)$ is given by equation (\ref{InitialTransitionFunctionTruncation}). Inserting the expression for $\phi_{0,s}(j+1,y)$  into
equation (\ref{QspTruncation}) we find
\begin{equation}\label{QCalculationT}
\begin{split}
\check{Q}_{s,p}(y)&=\frac{\Gamma(l-2n-\nu_1+1)\Gamma(\nu_1+1)}{\Gamma(l-2n+1)}\\
&\times\sum\limits_{j=0}^p\frac{(-1)^{p-j}}{(p-j)!j!}\frac{\Gamma(l-2n+j+1+p)}{\Gamma(\nu_1+1+j)}
G^{s,0}_{1,s}\left(\begin{array}{cccc}
                l-2n+j+1\\
               \nu_1+j,\nu_2,\ldots,\nu_s
             \end{array}\biggl|y\right).
\end{split}
\end{equation}
The Meijer G-function in the equation above has the following contour integral representation
\begin{equation}\label{GQCalculationT}
\begin{split}
&G^{s,0}_{1,s}\left(\begin{array}{cccc}
                l-2n+j+1\\
               \nu_1+j,\nu_2,\ldots,\nu_s
             \end{array}\biggl|y\right)\\
&=\frac{1}{2\pi i}
\int\limits_{c-i\infty}^{c+i\infty}\frac{\Gamma(u+\nu_1+j)\Gamma(u+\nu_2)\ldots\Gamma(u+\nu_s)}{\Gamma(u+l-2n+j+1)}y^{-u}du.
\end{split}
\end{equation}
We insert expression (\ref{GQCalculationT}) into formula (\ref{QCalculationT}), and interchange the sum and the integral. The resulting sum inside the integral can be written as
\begin{equation}
\begin{split}
&\sum\limits_{j=0}^p\frac{(-1)^{p-j}}{(p-j)!j!}\frac{\Gamma(u+\nu_1+j)\Gamma(l-2n+1+p+j)}{\Gamma(1+\nu_1+j)\Gamma(l-2n+1+u+j)}\\
&=(-1)^p\frac{\Gamma(u+\nu-1)\Gamma(l-2n+1+p)}{\Gamma(1+\nu-1)\Gamma(l-2n+1+u)}\frac{1}{p!}\;
{}_3F_2\left(-p,u+\nu_1,l-2n+1+p;1+\nu_1,l-2n+1+u;1\right).
\end{split}
\nonumber
\end{equation}
The Pfaff-Saalsch$\ddot{\mbox{u}}$ltz Theorem says that
$$
{}_3F_2\left(-p,a,b;c,d;1\right)=\frac{(c-a)_p(c-b)_p}{(c)_p(c-a-b)_p},
$$
if the balanced condition, $c+d=1-p+a+b$, is satisfied, see, for example, Ismail \cite{Ismail}, Section 1.4. In our case
$$
a=u+\nu_1,\;b=l-2n+1+p,\; c=1+\nu_1,\; d=l-2n+1+u,
$$
and the balanced condition is satisfied. Thus we have
$$
{}_3F_2\left(-p,u+\nu_1,l-2n+1+p;1+\nu_1,l-2n+1+u;1\right)=\frac{(1-u)_p(\nu_1-l+2n-p)_p}{(1+\nu_1)_p(2n-u-l-p)_p}.
$$
Taking into account that
$$
\frac{(1-u)_p}{(2n-l-u-p)_p}=\frac{(u-p)_p}{(u-2n+l+1)_p},
$$
we obtain the formula
\begin{equation}\label{QspTruncation1}
\begin{split}
\check{Q}_{s,p}(y)&=(-1)^p\frac{\Gamma(l-2n-\nu_1+1)}{\Gamma(l-2n+1)}\frac{\Gamma(l-2n+1+p)(\nu_1-l+2n-p)_p}{p!(1+\nu_1)_p}\\
&\times\frac{1}{2\pi i}\int\limits_{c-i\infty}^{c+i\infty}\frac{\Gamma(u+\nu_1)\ldots\Gamma(u+\nu_s)}{\Gamma(l-2n+1+u)}
\frac{(u-p)_p}{(u-2n+l+1)_p}y^{-u}du.
\end{split}
\end{equation}
 To write explicitly the second term in equation (\ref{MainGeneralFormula}) we insert  the formulae for $\check{P}_{r,p}$ (equation (\ref{PrpTruncation})),   and $\check{Q}_{s,p}$  (equation (\ref{QspTruncation1})) to equation
(\ref{KnmTruncated}).
This gives the formula for the correlation kernel in the statement of Theorem \ref{TheoremTruncation}. Note also that the second
term in the formula for the correlation kernel can be written as
\begin{equation}
\begin{split}
\widetilde{K}_{n,m}\left(r,x;s,y\right)=\frac{1}{(2\pi i)^2}
\int_{c-i\infty}^{c+i\infty}du\oint\limits_{\Sigma}\frac{\prod_{j=0}^s\Gamma(u+\nu_j+1)}{\prod_{j=0}^r\Gamma(t+\nu_j+1)}\\
\sum\limits_{p=0}^{n-1}(l-2n+2p+1)\frac{\Gamma(t-p)\Gamma(t+l-2n+p+1)}{\Gamma(u+l-2n+p+1)}x^ty^{-u}.
\end{split}
\end{equation}
The sum inside the integral is the same  as that in Kuijlaars,  Stivigny \cite{KuijlaarsStivigny}
(see the proof of Proposition 4.4 in Kuijlaars,  Stivigny \cite{KuijlaarsStivigny}), and
the rest of the proof of Proposition \ref{TheoremExactKernelTruncation} and Proposition \ref{PropositionScalingLimitTruncation} is the same as that of
Proposition 4.4 and Theorem 4.7 in Kuijlaars,  Stivigny \cite{KuijlaarsStivigny}.
\qed
\section{Proof of Theorem \ref{TheoremSourseExact} and Proposition \ref{PropositionScalingLimitSource}}
 When the initial matrix is  a Ginibre matrix with a source,
the Ginibre product process is a determinantal point process on $\left\{1,\ldots,m\right\}\times\R_{>0}$. This follows immediately form
Theorem \ref{TheoremPolynomialInitialConditions} and Proposition \ref{PropositionAPlusG}. Theorem \ref{TheoremPolynomialInitialConditions}
gives a general formula for the correlation kernel $K_{n,m}(r,x;s,y)$. The first term in  formula (\ref{MainGeneralFormula}) for $K_{n,m}(r,x;s,y)$ does not depend on initial conditions, and it is the same is in the previous cases. Let us derive a formula for the second term of $K_{n,m}(r,x;s,y)$.
By Theorem \ref{TheoremPolynomialInitialConditions} and Proposition \ref{PropositionAPlusG} the second term of $K_{n,m}(r,x;s,y)$ is given by
$$
\widetilde{K}_{n,m}\left(r,x;s,y\right)=\sum\limits_{i,j=1}^n\phi_{r,m+1}(x,i)\left(A^{-1}\right)_{i,j}\phi_{0,s}(j,y),
$$
where
$$
\phi_{r,m+1}(x,i)=\Gamma(\nu_{r+1}+i)\ldots\Gamma(\nu_m+i)x^{i-1},
$$
and
\begin{equation}\label{InitialTransitionGinibreSource}
\phi_{0,s}(j,y)=\frac{1}{\Gamma(\nu_1+1)}\int\limits_{0}^{\infty}{}_0F_1\left(\nu_1+1;q_jt\right)t^{\nu_1-1}e^{-t}
G^{s-1,0}_{0,s-1}\left(\begin{array}{ccc}
                - \\
               \nu_{2}\ldots,\nu_s
             \end{array}\biggl|\frac{y}{t}\right)dt.
\end{equation}
The matrix $A=\left(a_{i,j}\right)_{i,j=1}^n$
is given by
\begin{equation}
a_{i,j}=\Gamma(\nu_2+j)\ldots\Gamma(\nu_m+j)\int\limits_0^{\infty}{}_0F_1\left(\nu_1+1;q_it\right)e^{-t}t^{\nu_1+j-1}dt,
\nonumber
\end{equation}
where $i,j\in\left\{1,\ldots,n\right\}$. Since
\begin{equation}
L_k^{\nu_1}(y)=\frac{e^y}{k!\Gamma(\nu_1+1)}\int\limits_0^{\infty}t^{\nu_1+k}e^{-t}{}_0F_1\left(\nu_1+1;-yt\right)dt,
\nonumber
\end{equation}
we obtain
\begin{equation}
a_{i,j}=\prod_{j=0}^{m}\Gamma(\nu_0+j)L_{j-1}^{\nu_1}(-q_i)e^{q_i}.
\nonumber
\end{equation}
\begin{prop} Set $A^{-1}=\left(\gamma_{i,j}\right)_{i,j=1}^n$. This matrix is characterized by
\begin{equation}
\sum\limits_{j=1}^n\prod\limits_{l=0}^m\Gamma(\nu_l+j)L_{j-1}^{\nu_1}(u)\gamma_{j,k}=e^{-q_k}\underset{l\neq k}{\prod\limits_{l=1}^n}\frac{-u-q_l}{q_k-q_l},\;\; k\in\left\{1,\ldots,n\right\}.
\nonumber
\end{equation}
\end{prop}
\begin{proof}
See Forrester and  Liu \cite[Sec. 2.1]{ForresterLiu}.
\end{proof}
\begin{prop}
Set
\begin{equation}
\varphi_{r,m+1}(x,j)=\sum\limits_{i=1}^n\phi_{r,m+1}(x,i)\left(A^{-1}\right)_{i,j},\;\; j\in\left\{1,\ldots,n\right\}.
\end{equation}
The function $\varphi_{r,m+1}(x,j)$ can be written as
\begin{equation}
\varphi_{r,m+1}(x,j)=\int\limits_0^{\infty}u^{\nu_1}e^{-u}\Psi_r(u;x)e^{-q_j}\underset{l\neq j}{\prod\limits_{l=1}^n}\frac{-u-q_l}{q_j-q_l}du,
\end{equation}
where
\begin{equation}
\begin{split}
\Psi_r(u;x)=\frac{1}{(2\pi i)^r\Gamma(\nu_1+1)}\int\limits_{\gamma_1}&dw_1\ldots\int\limits_{\gamma_r}dw_r
\prod\limits_{l=1}^rw_l^{-\nu_l-1}e^{w_l}\\
&\times e^{\frac{x}{w_1\ldots w_r}}{}_0F_1\left(\nu_1+1;-\frac{xu}{w_1\ldots w_r}\right).
\end{split}
\end{equation}
In the formula above  $\gamma_1$, $\ldots$, $\gamma_r$ are paths starting and ending at $-\infty$ and encircling the origin
once in the positive direction.
\end{prop}
\begin{proof} See Forrester and Liu \cite[Sec. 2.1]{ForresterLiu}.
\end{proof}
\begin{prop} The initial transition function $\phi_{0,s}(j,y)$ can be written as
$$
\phi_{0,s}(j,y)=\Phi_s(-q_j;y),
$$
where the auxiliary function $\Phi_s(q;y)$ is defined by
$$
\Phi_s(q;y)=\frac{1}{2\pi i}\int\limits_{c-\i\infty}^{c+i\infty}dzy^{-z}\vartheta(q;z)\Gamma(\nu_2+z)\ldots\Gamma(\nu_s+z).
$$
In the formula above $c>-\mbox{min}\left(\nu_1,\ldots,\nu_s\right)$, $s\in\left\{1,\ldots,m\right\}$, and
$$
\vartheta(q;z)=\frac{1}{\Gamma(\nu_1+1)}\int\limits_0^{\infty}dtt^{\nu_1+z-1}e^{-t}{}_0F_1\left(\nu_1+1;-qt\right).
$$
\end{prop}
\begin{proof}
Use equation (\ref{InitialTransitionGinibreSource}), and the integral representation for the Meijer $G$-function
inside the integral in equation (\ref{InitialTransitionGinibreSource}),
\begin{equation}
G^{s-1,0}_{0,s-1}\left(\begin{array}{ccc}
                - \\
               \nu_{2},\ldots,\nu_s
             \end{array}\biggl|\frac{y}{t}\right)=\frac{1}{2\pi i}\int\limits_{c-i\infty}^{c+i\infty}\gamma(\nu_2+z)\ldots\Gamma(\nu_s+z)
\left(\frac{y}{t}\right)^{-z}dz.
\end{equation}
\end{proof}
Now  we have
\begin{equation}
\begin{split}
&\widetilde{K}_{n,m}\left(r,x;s,y\right)=\sum\limits_{j=1}^n\varphi_{r,m+1}(x,j)\phi_{0,s}(j,y)\\
&=\int\limits_0^{\infty}u^{\nu_1}e^{-u}\psi_r(u;x)\sum\limits_{j=1}^ne^{-q_j}\underset{l\neq j}{\prod\limits_{l=1}^{n}}\Phi_s\left(-q_j,y\right)du.
\end{split}
\end{equation}
We apply the same argument as in  Forrester and Liu \cite[Sec. 2.1]{ForresterLiu},
and by the residue theorem we obtain a kernel equal to that defined by equation  (\ref{knmsourse}).
Theorem \ref{TheoremSourseExact} is proved. The scaling limits of the correlation kernel in different asymptotic regimes
given by Proposition  \ref{PropositionAPlusG}  are derived  by  a repetition of arguments from Forrester and Liu \cite[Sec. 3.1]{ForresterLiu}.
\qed

\end{document}